\documentclass[a4paper,11pt,,english]{article}
\usepackage[T1]{fontenc}
\usepackage{hyperref}
\usepackage{verbatim}
\usepackage{amsfonts}
\usepackage{amsmath,amssymb,amsthm}
\usepackage{url}
\usepackage{fullpage}
\usepackage[usenames]{color}
\usepackage{times}
\usepackage{babel}
\usepackage{enumitem}

\newcounter{method1}
\newtheorem{meth}[method1]{Construction}
\newtheorem{theorem}{Theorem}[section]
\newtheorem{lemma}[theorem]{Lemma}
\newtheorem{proposition}[theorem]{Proposition}

\newtheorem{fact}[theorem]{Fact}
\newtheorem{corollary}[theorem]{Corollary}

\newtheorem{clam}[theorem]{Claim}
\newtheorem{obs}[theorem]{Observation}

\newtheorem{definition}[theorem]{Definition}
\newtheorem{rem}[theorem]{Remark}

\newcommand{\eps}{\epsilon}

\newcommand{\N}{{\mathbb N}}

\newcommand{\polylog}{\operatorname{polylog}}
\newcommand{\loglog}{\operatorname{loglog}}
\newcommand{\poly}{\operatorname{poly}}

\newcommand{\dsum}{\displaystyle\sum}
\newcommand{\expect}{\mathbb{E}}

\newcommand{\mbN}{\mathbb{N}}
\newcommand{\mcA}{{\mathcal A}}
\newcommand{\mcLA}{{\mathcal L \mathcal A}}

\newcommand{\mcP}{{\mathcal P}}

\newcommand{\mcR}{{\mathcal R}}

\newcommand{\mcH}{{\mathcal H}}

\newcommand{\sd}{\leq_{st}}

\newcommand{\ignore}[1]{}

\makeatother

\begin{document}

\date{}

\title{New Techniques and Tighter Bounds for Local Computation Algorithms}
\author{ Omer Reingold \thanks{Microsoft Research. E-mail: {\tt  omer.reingold@microsoft.com} } \and Shai Vardi
\thanks{School of Computer Science, Tel Aviv University.
 E-mail: {\tt  shaivar1@post.tau.ac.il}. This research was supported in part by the Google Europe Fellowship in Game Theory.}
 }

\maketitle
\begin{abstract}
Given an input $x$, and a search problem $F$, local computation algorithms (LCAs) implement access to specified locations of $y$ in a legal output $y \in F(x)$, using polylogarithmic time and space. Mansour et al., (2012), had previously shown how to convert certain online algorithms to LCAs.  In this work, we expand on that line of work and develop new techniques for designing LCAs and bounding their space and time complexity. 
Our contributions are fourfold: (1) We significantly improve the running times and space requirements of LCAs for previous results, (2) we expand and better define the family of online algorithms which can be converted to LCAs using our techniques, (3) we show that our results apply to a larger family of graphs than that of previous results, and  (4) our proofs are simpler and  more concise than the previous proof methods. 

For example,  we show how to construct LCAs that require  $O(\log{n}\log\log{n})$ space and $O(\log^2{n})$ time (and expected time $O(\log\log n)$) for problems such as maximal matching on a large family of graphs, as opposed to the henceforth best results that required   $O(\log^3{n})$ space and $O(\log^4{n})$ time, and applied to a smaller family of graphs. 
\end{abstract}

\section{Introduction}

The need for \emph{local computation algorithms} (LCAs) arises in situations when we require fast and space-efficient access to part of a solution to a computational problem, but we never need the entire solution at once. Consider, for instance, a network with millions of nodes, on which we would like to compute a maximal independent set. We are never required to compute the entire solution; instead,  we are inquired\footnote{To avoid confusion, we use the word \emph{inquire} to represent queries made \emph{to} the algorithm, and the word \emph{query} to refer to the queries the algorithm itself makes to the data.} about specific nodes, and we have to reply whether or not they are part of the independent set. To accomplish this, we are allocated space polylogarithmic in the size of the graph, and are required to reply to each inquiry in polylogarithmic time. Although at any single point in time, we can be inquired as to whether a single node (or a small number of nodes) is in the solution, over a longer time period, we may be inquired about all of the nodes. We therefore  want all of our replies to be consistent with the same maximal independent set. 


Local computation algorithms were introduced by Rubinfeld, Tamir, Vardi and Xie, \cite{RTVX11}. In that work, they gave LCAs for several problems such as hypergraph $2$-coloring and maximal independent set; their focus was on the time bounds of these algorithms. To give better space bounds, Alon, Rubinfeld, Vardi and Xie, \cite{ARVX11}  showed
how to use $k$-wise independent hash functions (instead of random functions), when an LCA makes at most $k$ queries (per inquiry), and used it to obtain polylogarithmic time and space bounds for hypergraph $2$-coloring. 
Mansour, Rubinstein, Vardi and Xie \cite{MRVX12}, showed how to bound the number of queries of the algorithms of \cite{ARVX11} by $O(\log{n})$ (where $n$ is the number of vertices), and showed that this query bound also applies to a variety of LCAs which are obtained via reductions from online algorithms. Their results give, for example, an LCA for maximal matching which requires $O(\log^3{n})$ space and runs in $O(\log^4{n})$ time. 
 The analyses of \cite{ARVX11} and \cite{MRVX12} (and indeed of many other works, e.g., \cite{HKNO09,MV13}) use an abstract structure called a \emph{query tree}, introduced by Nguyen and Onak \cite{NO08}, to bound the number of queries the LCA makes.

In this work, we depart from that line of proof, and develop new ideas for  bounding the space and time bounds of LCAs. Our techniques are more general and can be applied to a wider family of graphs. The proofs of these new techniques are simpler, and this allows us to improve on the  known results on LCAs in several important ways: we reduce the space and time bounds of most of the previous results; we expand and better characterize  the family of LCAs for which these new bounds hold; and we expand the family of graphs on which these LCAs can be implemented.

\subsection{Our Contributions}
We formalize a family of algorithms called \emph{neighborhood-dependent online algorithms}. These algorithms compute some function $f$ on the vertices (or edges) of a graph, and the value of $f(x)$ for every vertex $x$ depends only on the value of $f$ for the immediate neighbors of $x$ that arrived before $x$. This family includes the algorithms of \cite{MRVX12} and, furthermore, encompasses a large number of online algorithms: online algorithms for maximal matching,  maximal independent set, and vertex/edge  coloring, for example. It also includes many online load balancing algorithms, such as the algorithm of Azar, Broder, Karlin and Upfal  \cite{ABK+99} and many of the variations thereof (e.g.,  \cite{Vocking03,Wie07,BBFN14}). 

 The results on LCAs thus-far have been on a very narrow family of graphs (for example, the LCAs of \cite{RTVX11} and \cite{ARVX11} only apply to graphs of bounded constant degree). We take an important step forward in characterizing the graphs which admit LCAs: we introduce a family of graphs, which we call \emph{$d$-light} graphs, which captures a broad range of graphs, for example:
\begin{itemize}[noitemsep,nolistsep]
\item $d$-regular graphs, or more generally, graphs with degree bounded by $d$, where $d=O(\log\log{n})$,
\item The random graphs $G(n,p)$, where $p=d/n$ for any $d = O(\log\log{n})$,
\item Bipartite graphs on $n$ consumers and $m$ producers, where each consumer is connected to $d$ producers at random. The degree of each consumer is fixed whereas the degree of each producer is distributed according to the binomial distribution $B(n, d/m)$.
\end{itemize}

Our main result is the following: Assume we are given a $d$-light graph, $G=(V,E)$ (for constant $d$), where $|V|=n$. We are also given a combinatorial search problem $F$ for $G$, for which there exists a (sufficiently efficient) neighborhood-dependent online algorithm. We show how to construct  the following:
\begin{enumerate}
\item An LCA for $F$ which requires time $O(\log{n}\log\log{n})$ and space $O(\log^2{n})$ per inquiry, with probability $1-1/\poly(n)$.
\item An LCA for $F$ which requires time $O(\log^2{n}) $ and space $O(\log{n}\log\log{n})$ per inquiry, with probability $1-1/\poly(n)$.
\end{enumerate}
We further show, that in both of these cases, the LCAs require $O(\log\log{n})$ time and $O(\log{n})$ space in expectation. 

For both of our LCAs we use the following general construction. When we are inquired about the value of some vertex $v$ in the solution,  we simulate the execution of the online algorithm on the graph, for some randomized order of arrival of the vertices. Because the algorithm is neighborhood-dependent, we only need to query the neighbors of $v$ that arrived before it. Call this set $U$. For each of the vertices in $u \in U$, we need to query the neighbors of $u$ that arrived before $u$, and so on. We show how to generate a pseudorandom ordering on the vertices using a seed of length $O(\log{n})$ such that  any such inquiry will require at most $O(\log{n})$ queries to the graph w.h.p. We furthermore show that using this seed, the expected number of queries the LCA will need to make to the graph is constant. The difference between the two constructions is the following: as we discover the vertices we need to query in order to simulate the online algorithm, we relabel the vertices with unique labels to reflect the ordering. Trivially, we can use labels of length $O(\log{n})$, and hence we get the first construction. However, as we only need to order $O(\log{n})$ vertices per inquiry, we only require labels of length $O(\log\log{n})$. Our second construction shows that we can create shorter labels, albeit at the expense of a somewhat worse running time. 
We show that if we allow our LCAs use polylogarithmic space and require polylogarithmic  time, we can handle  $d$-light graphs, where $d$ is as large as $\Theta(\log\log{n}) $. We also show that this bound is tight, at least for LCAs constructed by using the current general techniques of simulating online algorithms. 

In the majority of the LCAs that have been considered thus far, because the number of queries per inquiry is $O(\log{n})$, it is enough that the ordering of the vertices is $O(\log{n})$-wise independent.
 In some cases, though, we can implement LCAs which only require $O(1)$-wise independence. For example, the maximal independent set LCA of \cite{ARVX11} requires $O(1)$-wise independence, and they show an LCA which requires $O(\log^2{n})$ space and $O(\log^3{n})$ time. In addition to the construction above, we show a simpler construction, which is useful for improving the running time and space of LCAs when only $O(1)$-wise independence is required. For example, this construction can be used to reduce the running time and space of the LCA of \cite{ARVX11} to $O(\log{n})$. 

In addition, as a by-product of our main results, we make some interesting observations about stochastic dominance, and in particular, stochastic dominance of certain binomial distributions.
\subsection{Related Work}

In the face of ever-growing networks and databases, the last decade or so has produced a large volume of work on handling the challenges posed by these giant entities. While in the past, much of the research in algorithms had focused on what can be done in polynomial time, for these new networks, even linear-time tractability  is often insufficient.  Many of the approaches to this problem therefore concern themselves with what can be done in sublinear time. The approaches have been very diverse.

The field of property testing asks whether some mathematical object, such as a graph, has a certain property, or is ``far'' from having that property (for an introduction and a recent survey, see \cite{Goldreich11} and \cite{Ron09} respectively). Streaming algorithms are required to use limited memory and quickly process data streams which are presented as a sequence of items (see \cite{Muth05} for a comprehensive introduction, and \cite{Zhang10} for a recent survey).  Sublinear approximation algorithms give approximate solutions to optimization problems, (e.g., \cite{PR07,NO08,RS11}). A major difference between all of these algorithms and LCAs is that LCAs require that w.h.p., the output will be correct on any input, while optimization
problems usually require a correct output only on most inputs. More importantly, LCAs require a consistent output for each inquiry, rather than only approximating a given global property. 

\emph{Locally decodable codes} (LDCs) \cite{KT00} which given an encoding of a message, provide quick access to the requested bits of the original message, can be viewed as LCAs.   Known constructions of LDCs are efficient and use small space \cite{Yek12}.   LCAs also generalize the \emph{reconstruction} models described in \cite{ACC+08, SS10, JR11}.  These models describe scenarios where an input string that has a  certain property, such as monotonicity, is assumed to be corrupted at a relatively small number of
locations.  The reconstruction algorithm gives fast query access to an uncorrupted version of the string that is close to the original input.  
Another related area of research is distributed algorithms which run in time independent (or almost independent) of the size of the network, often dependent on the maximal degree, which is assumed to be small (e.g., \cite{Linial92, NS95} for some classic results and \cite{BEPS12, PS13} for newer results). Recently, much focus has been given to constant time distributed algorithms (for a recent survey, see \cite{Suomela13}). These algorithms have a nice relationship to LCAs, in that any distributed algorithm which runs in constant time immediately gives an LCA for the problem, when the degree is constant. In fact, following the results Hassidim, Mansour and Vardi \cite{HMV14} and the results of this paper,  it is easy to see that this holds for any $d$-light graphs.

\subsection{Paper Organization}

In Section \ref{section:prelim} we define the terms we need for our results. We also prove some simple results on stochastic dominance and $k$-wise independence. In Section \ref{section:delight} we define $d$-light graphs, and show that the neighborhood of any suitably defined vertex set is not too large in such a graph. In Section \ref{section:orderings} we show a simple construction of a randomized ordering, which improves the space and time bounds of previous results. We show that we can make even greater improvements: Section \ref{subsection:better} is the heart of our main result - we show  how to generate a randomized ordering on the vertices which guarantees a fast running time for our local computation algorithms using a logarithmic seed. In Section \ref{section:expected}, we show that the expected number of vertices we need to query per inquiry is constant. In  Section \ref{section:onlinetolca}, we show how to use the results of Sections \ref{subsection:better} and \ref{section:expected} to obtain fast and space-efficient LCAs.
\section{Preliminaries}
\label{section:prelim}


We denote the set $\{1, 2, \ldots n\}$ by $[n]$. Let $X$ be a random variable, distributed according to the binomial distribution with parameters $n$ and $p$;  we denote this by $X \sim B(n,p)$.
We assume the standard uniform-cost RAM model, in which  the word size $w = O(\log{n})$, and it takes $O(1)$ to read and perform simple operations on words.

Let $G=(V,E)$ be a graph. The neighborhood of a vertex $v$, denoted $N(v)$, is the set of vertices that share an edge with $v$: $N(v)=\{u:(u,v)\in E\}$. The \emph{degree} of a vertex $v$ is $|N(v)|$. The neighborhood of a set of vertices $S \subseteq V$, denoted $N(S)$ is the set of all vertices $\{u : v \in S, u \in N(v) \setminus S\}$.

\subsection{Local Computation Algorithms}
We use the model of local computation algorithms (LCAs) of \cite{MV13} (a slight variation of the model of \cite{RTVX11}).

\begin{definition}[Local computation algorithms] A {\em $(t(n),$ $ s(n), $ $\delta(n))$ - local computation
algorithm} $\mcLA$ for a computational problem is a (randomized)
algorithm that receives an input of size $n$, and an inquiry $x$.
Algorithm $\mcLA$ uses at most $s(n)$ memory and replies to $x$ in time $t(n)$,  with probability
at least $1-\delta(n)$.
The algorithm must be \emph{consistent}, that is, for any fixing of its randomness, the algorithm's replies to all of the
possible inquiries combine to a single feasible solution to the
problem.
\end{definition}
\begin{rem}Usually, LCAs  require a random seed that is stored in memory and reused for every inquiry. In these cases, the space requirement of the LCA includes the space required to store this seed as well as the space to perform the computations of the algorithm.
\end{rem}

In this paper, we give LCAs for computational problems solvable by online algorithms. For simplicity, we only consider online graph algorithms on vertices; an analogous definition holds for algorithms on edges, and our results hold for them as well. Intuitively, such an algorithm is presented with the vertices of a graph $G=(V,E)$ in some arbitrary order. Once the algorithm is presented with a vertex $v$ (as well as the edges to the neighbors of $v$ that arrived before $v$), it must irrevocably output a value which we denote $f(v)$.  The output of the algorithm is the combination of all of these intermediate outputs, namely the function $f: V \rightarrow O$ (where $O$ is some arbitrary finite set).

We require that the online algorithm will also be neighborhood dependent in the sense that the value $f(v)$ is only a function of the values $\{f(u)\}$ for the neighbors $u$ of $v$ which the algorithm has already seen. 
Formally,
\begin{definition}[Neighborhood-dependent online graph algorithm]\label{def:NDonline}
Let $G=(V,E)$ be a graph, and let $O$ be some arbitrary finite set.
A \emph{neighborhood-dependent online graph algorithm} $\cal A$ takes as input a vertex $v \in V$ and a sequence of pairs $\{(u_1,o_1),\ldots(u_\ell,o_\ell)\}$ where  $\forall i, u_i\in V, o_i\in O$, and outputs a value in $O$. For every permutation $\Pi$ of the vertices in $V$, define the output of $\cal A$ on $G$ with respect to $\Pi$ as follows.
Denote by $v_i$ the vertex at location $i$ under $\Pi$. Define $f(v_i)$ recursively by invoking $\mcA$ on $v_i$ and the sequence of values $(v_j,f(v_j))$, such that $j<i$ and $(v_j,v_{i})\in E$.

Let $R$ be a search problem on graphs. 
We say that $\cal A$ is a neighborhood-dependent online graph algorithm for $R$ if for every graph $G$ and every permutation $\Pi$, the output of $\cal A$ on $G$ with respect to $\Pi$ satisfies the relation defining $R$.\\
\end{definition}

\subsection{Stochastic Dominance}
\label{section:stoch}
We consider graphs with degree distribution bounded by the binomial distribution. We rely on stochastic dominance.
\begin{definition}[Stochastic dominance]
For any two distributions over the reals, $X$ and $Y$, we say that $X$ is \emph{stochastically dominated by}\footnote{This is usually called \emph{first-order stochastic dominance}. As this is the only measure of stochastic dominance we use, we omit the term ``first-order'' for brevity.} $Y$ if for every real number $x$ it holds that $\Pr[X> x]\leq \Pr[Y> x]$. We denote this by $X \leq_{st} Y.$
\end{definition}


We recall the following facts about stochastic dominance (see, e.g., \cite{SS94}).
\begin{enumerate}\label{fact:sdexpect} 
\item If $X \sd Y$ and $Y \sd Z$ then $X \sd Z$.
\item For any integer $n$, let $\{X_1, X_2, \ldots\ X_n\}$ and $\{Y_1, Y_2,\ldots Y_n\}$ be two sequences of independent random variables. If $\forall i, X_i \sd Y_i$, then
$$\dsum_{j=1}^{n} X_j \sd \dsum_{j=1}^{n}Y_j.$$
\end{enumerate}

We need the following lemmas, which we prove in Appendix \ref{app:lrst}. 
\begin{lemma} \label{lemma:stoch2}
Let $\{Y_1, Y_2, \ldots, Y_n\}$  be a sequence of independent random variables. Let  $\{X_1, X_2,\ldots, X_n\}$ be a series of (possibly) dependent random variables. If it holds that $X_1 \sd Y_1$, and for any $1\leq i \leq n$, conditioned on any realization of $X_1,\ldots X_{i-1}$, it holds that $X_i \sd Y_i$, then
 $$\dsum_{j=1}^{n} X_j \sd \dsum_{j=1}^{n}Y_j.$$
\end{lemma}

\begin{lemma}\label{lemma:lrst}
Let  $Z$ and $X$ be random variables such that $Z \sim 2d + B(n^2, \frac{2d}{n^2})$ and $X \sim B(\alpha, \frac{d}{\alpha})$, where $d \leq \alpha \leq n$. Then $X \sd Z$.
\end{lemma}

 \subsection{Static and Adaptive $k$-wise Independence}
 \begin{definition}[$k$-wise independent hash functions] For $n,L,k \in \N$ such that $k \leq n$, a
 family of functions $\mcH = \{h : [n] \rightarrow [L]\}$ is \emph{$k$-wise independent} if for all distinct $x_1, x_2, \ldots ,x_k \in [n]$, when
 $H$ is sampled uniformly from $\mcH$ we have that the random variables $H(x_1), H(x_2), \ldots ,H(x_k)$ are independent and uniformly distributed in $[L]$.
 \end{definition}

 To quantify what we mean by ``almost'' $k$-wise independence, we use the notion of \emph{statistical distance}.
 \begin{definition}[Statistical distance]
  For random variables $X$ and $Y$ taking values in $\mathcal{U}$, their \emph{statistical distance}  is $$\Delta(X, Y) = max_{D\subset \mathcal{U}}|\Pr[X \in D] - \Pr[Y \in D]|.$$ For $\eps\geq 0$, we say that $X$ and
 $Y$ are $\eps$-close if $\Delta(X, Y)\leq \eps$.
 \end{definition}

 \begin{definition}[$\eps$-almost $k$-wise independent hash functions] For $n,L,k \in \N$ such that $k \leq n$, let $Y$ be a random variable sampled uniformly at random from $[L]^k$. For $\eps\geq 0$, a family of functions $\mcH = \{h : [n] \rightarrow [L]\}$ is \emph{$\eps$-almost $k$-wise independent} if for all distinct $x_1, x_2, \ldots ,x_k \in [n]$, we have that $\langle H(x_1), H(x_2),\ldots ,H(x_k)\rangle $ and $Y$ are $\eps$-close,  when
 $H$ is sampled uniformly from $\mcH$.  We sometimes use the term ``$\eps$-dependent'' instead of ``$\eps$-almost independent''.
 \end{definition}

Another interpretation of $\eps$-almost $k$-wise independent hash functions is as functions that are indistinguishable from uniform for a static distinguisher that is allowed to query the function in at most $k$ places. In other words, we can imagine the following game: the (computationally unbounded) distinguisher $\cal D$ selects $k$ inputs $x_1, x_2, \ldots ,x_k \in [n]$, and gets in return values $F(x_1), F(x_2),\ldots ,F(x_k)$ where $F: [n] \rightarrow [L]$ is either chosen from $\mcH$ or is selected uniformly at random. $\mcH$ is $\eps$-almost $k$-wise independent if no such $\cal D$ can differentiate the two cases with advantage larger than $\epsilon$. In this paper we will need to consider {\em adaptive} distinguishers that can select each $x_i$ based on the values $F(x_1), F(x_2),\ldots ,F(x_{i-1})$.
 \begin{definition}[$\eps$-almost adaptive $k$-wise independent hash functions] \label{def:adaptiveKwise} For $n,L,k \in \N$ such that $k \leq n$ and for $\eps\geq 0$, a
 family of functions $\mcH = \{h : [n] \rightarrow [L]\}$ is \emph{$\eps$-almost adaptive $k$-wise independent} if for every (computationally unbounded) distinguisher $\cal D$ that makes at most $k$ queries to an oracle $F$ it holds that $$|\Pr[{\cal D}^H(1^n)=1] - \Pr[{\cal D}^G(1^n)=1] |\leq \eps$$  when
 $H$ is sampled uniformly from $\mcH$ and $G: [n] \rightarrow [L]$ is  selected uniformly at random.

 We say that $\mcH$ is \emph{adaptive $k$-wise independent} if it is $0$-almost adaptive $k$-wise independent.
 \end{definition}

Maurer and Pietrzak \cite{MP04} showed a very efficient way to transform a family of (static) $k$-wise almost independent functions into a family of adaptive $k$-wise almost independent functions with similar parameters. For our purposes, it is enough to note that every family of (static) $k$-wise almost independent function is in itself also adaptive $k$-wise almost independent. While the parameters deteriorate under this reduction, they are still good enough for our purposes. We provide the reduction here for completeness.

\begin{lemma}\label{lem:static2adaptive} 
For $n,L,k \in \N$ such that $k \leq n$ and for $\eps\geq 0$, every family of functions $\mcH = \{h : [n] \rightarrow [L]\}$ that is $\eps$-almost $k$-wise independent is also $\eps L^k$-almost adaptive $k$-wise independent.
\end{lemma}
\begin{proof}
The proof is by a simulation argument. Consider an adaptive distinguisher $\cal D$ that makes at most $k$ queries; assume without loss of generality that $\cal D$ always makes {\em exactly} $k$ distinct queries. We can define the following static distinguisher $\cal D'$ with oracle access to some function $F$ as follows: $\cal D'$ samples $k$ distinct outputs $y_1, y_2, \ldots ,y_k \in [L]$ uniformly at random. $\cal D'$ then simulates $\cal D$ by answering the $i$th query $x_i$ of $\cal D$ with $y_i$. Let $\sigma$ be the bit $\cal D$ would have output in this simulation. Now $\cal D'$  queries $F$ for $x_1,x_2,\ldots x_k$ (note that $\cal D'$ makes all of its queries simultaneously and is therefore static). If the replies obtained are consistent with the simulation (i.e.\  $y_i=F(x_i)$ for every $i$) then $\cal D'$ outputs $\sigma$. Otherwise, it outputs $0$. By definition, $$\Pr[{\cal D'}^F(1^n)=1] = L^{-k}\Pr[{\cal D}^F(1^n)=1].$$
This implies that for every $H$ and $G$
$$|\Pr[{\cal D'}^H(1^n)=1] - \Pr[{\cal D'}^G(1^n)=1]| = L^{-k}|\Pr[{\cal D}^H(1^n)=1] - \Pr[{\cal D}^G(1^n)=1]|.$$  The proof follows. 
\end{proof}

In particular, Lemma~\ref{lem:static2adaptive} implies that $k$-wise independent functions are also adaptive $k$-wise independent, and we get the following theorem:

\begin{theorem}[cf.\ \cite{Vad12} Proposition~3.33 and Lemma~\ref{lem:static2adaptive}]\label{thm:simple}
For $n,k \in \N$ such that $k \leq n$ and $n$ is a power of $2$, there exists a family of functions $\mcH =  \{h : n \rightarrow n \}$ that is adaptive $k$-wise independent, whose seed length is $k \log{n}$. The time required to evaluate each $h$ is $O(k)$ word operations and the  memory required per evaluation is $O(\log{n})$ bits (in addition to the memory for the seed).
\end{theorem}

 Naor and Naor \cite{NN90} showed that relaxing from $k$-wise to almost $k$-wise independence can imply significant savings in the family size. As we we also care about the evaluation time of the functions, we will employ a very recent result of Meka, Reingold, Rothblum and Rothblum \cite{MRRR13} (which, using Lemma~\ref{lem:static2adaptive}, also applies to {\em adaptive} almost $k$-wise independence). The following is specialized from their work to the parameters we mostly care about in this work.

 \begin{theorem} [\cite{MRRR13} and Lemma~\ref{lem:static2adaptive}]\label{thm:almost}
 For every $n,L,k \in \N$ and $\eps > 0$ such that $L$ is a power of $2$, and such that $k\cdot L=O(\log n)$ and $1/\eps=poly(n)$, there is a family of $\eps$-almost adaptive $k$-wise  independent functions
 $\mcH =  \{h : n \rightarrow L \}$  such that choosing a random function from $\mcH$ takes $O(\log{n})$ random bits. The time required to evaluate each $h$ is $O(\log k)$ word operations and the memory is $O(\log{n})$ bits.
 \end{theorem}

The property of $k$-wise almost independent hash functions $H$ that we will need is that an algorithm querying $H$ will not query ``too many" preimages of any particular output of $H$. Specifically, if the algorithm queries $H$ more than $cL \log{n}$ times, none of the values will appear more than twice their expected number. More formally:
\begin{proposition}
\label{prop:1}
There exists a constant $c$ such that following holds. Let $n,L,k \in \N$ be such that $k \leq n$ and let  $\eps\geq 0$. Let $\mcH = \{h : [n] \rightarrow [L]\}$ be a family of functions that is $\eps$-almost adaptive $k$-wise independent. Let $\cal A$ be any procedure with oracle access to $H$ sampled uniformly from $\mcH$ and let $\ell$ be any value in $[L]$. Define the random variable $m$ to be the number of queries $\cal A$ makes and the random variables $x_1,x_2,\ldots, x_{m}$ to be those queries. Then conditioned on $c L \log n \leq m\leq k$, it holds that:
$$\Pr[|\{x_i| H(x_i)=\ell\}| > 2m/L] \leq \frac{1}{2n^5}+\eps.$$
\end{proposition}

\begin{proof}
Consider any $\cal A$ as in the theorem and assume for the sake of contradiction that $\Pr[|\{x_i| H(x_i)=\ell\}| > 2m/L\ |\ c L \log n \leq m\leq k]> \frac{1}{2n^5}+\eps.$ Consider $\cal A$ with access to a uniformly selected $G : [n] \rightarrow [L]$. Define $m'$ to be the number of queries $\cal A$ makes in such a case and let $x'_1,\ldots, x'_m$ be the set of queries. By the Chernoff bound, conditioned on any fixing of $m'$, we have that $\Pr[|\{x'_i| H(x'_i)=\ell\}| > 2m'/L]$ is exponentially small in $m'/L$. Therefore, by a union bound, for a sufficiently large $c$ we have that $\Pr[|\{x_i| H(x_i)=\ell\}| > 2m'/L\ |\ c L \log n \leq m'\leq k]\leq \frac{1}{2n^5}$ (note that $L<n$, otherwise the theorem is trivially true). We thus have that $\cal A$ distinguishes the distribution $H$ from the distribution $G$ with $k$ queries,  with advantage $\eps$, in contradiction to $H$ being $\eps$-almost adaptive $k$-wise independent. 
\end{proof}

\section{Exposures in $d$-light Graphs}
\label{section:delight}
In this section we introduce the family of graphs to which our algorithms apply. We only discuss {\em undirected graphs,} but both the definitions and algorithms easily extend to the directed case. We call these graphs \emph{$d$-light} and they generalize a large family of graphs (see below). Before defining $d$-light graphs, we need to define certain processes that adaptively expose vertices of a graph.

\begin{definition}[Adaptive vertex exposure]
An \emph{adaptive vertex exposure} process $\mcP$ is a process which receives a limited oracle access to a graph $G=(V,E)$ in the following sense: $\mcP$ maintains a set of vertices $S \subseteq V$, (initially $S=\emptyset$), and updates $S$ iteratively. In the first iteration, $\mcP$ exposes an arbitrary vertex $v \in V$, and adds $v$ to $S$. In each subsequent iteration,  $\mcP$ can choose to expose any $v \in N(S)$: the vertex is added to $S$, and $\mcP$ learns the entire neighborhood of $v$. If a subset $S \subseteq V$ was exposed by such a process, we say that $S$ was \emph{adaptively exposed}.
\end{definition}

\begin{definition}[$d$-light distribution]
Let $d>0$ and let $G=(V,E)$ be a distribution on graphs. If, for every adaptively exposed $S \subset V$, conditioned on any instantiation of $S\cup N(S)$ and of the set of edges $E_S=\{(u,v)\in E\ |\ u\in S\}$ we have that for every $v \in N(S)\setminus S$ (or any $v$ in the case that $S$ is empty), there is a value $d \leq \alpha \leq |V|$, such that $|N(v)\setminus S|$ is stochastically dominated by $B(\alpha, d/\alpha)$, we say that the degree distribution of $G$ is \emph{$d$-light} (or, for simplicity, that $G$ is $d$-light).
\end{definition}

The family of $d$-light graphs includes many well-studied graphs in the literature; for example,
\begin{itemize}[noitemsep,nolistsep]
\item $d$-regular graphs, or more generally, graphs with degree bounded by $d$ (taking $\alpha = d$),
\item The random graphs $G(n,p)$, where $p=d/(n-1)$. Each one of the ${n \choose 2}$ edges is selected independently with probability $p$; therefore the degree of each vertex is distributed according to the binomial distribution $B(n-1, d/(n-1))$.  Note that in general, the degrees are not independent (for example, if one vertex has degree $n-1$ then all other vertices are connected to this vertex). However, once a subset $W \subset V$ has been exposed (as well as the edges in the cut $(W, G\setminus W)$), for any $v \notin W$, $|N(v)\setminus W| \sim B(n-|W|-1, d/(n-1))$, which is stochastically dominated by $B(n-1, d/(n-1))$.
\item Some graphs where the degrees of different vertices are distributed differently. For example, consider a bipartite graph on $n$ consumers and $m$ producers. Assume that each consumer is connected to $d$ random producers. The degree of each consumer is fixed whereas the degree of each producer is distributed according to the binomial distribution $B(n, d/m)$.
\end{itemize}

\begin{rem}
Our algorithms can work even with graphs satisfying weaker definitions than being $d$-light. In particular, the exposure procedure we define in our proof is computationally efficient (whereas the definition allows for computationally unbounded exposure), and it exposes at most logarithmically many vertices. One of the other possible relaxations is allowing for more general distributions than binomial (as long as the sum of degrees satisfies strong enough tail inequalities). We use the less general definition for the sake of simplicity (and at the expense of generality).
\end{rem}

\subsection{Bounding the Neighborhood of Exposed Sets}

The following proposition proves the property of $d$-light graphs that is needed for our analysis.  For motivation, consider the following property: every connected subgraph with $s$ vertices, where $s$ is at least logarithmic in the number of vertices, has at most $O(d\cdot s)$ neighbors. (Note that this property holds trivially for $d$-regular graphs.) In Proposition~\ref{prop:neighbors} we ask for the weaker property that when we adaptively expose a (large enough) connected subgraph, it is unlikely to have too many neighbors. So while "high-degree subgraphs", may possibly exist, it is unlikely that they will be exposed.
\begin{proposition} \label{prop:neighbors}
For some constant $c>0$, for every $d$-light graph  $G=(V,E)$ with $|V|=n$ and every adaptively exposed subset of the vertices $S\subseteq V$ of size at least $c \log{n}$, we have that $\Pr[|\{(u,v)\in E\ |\ u\in S\}|> 6d|S|]\leq 1/n^5$. In particular, $\Pr[ | N(S)| > (6d-1)|S|] \leq 1/n^5$.
\end{proposition}

\begin{proof}
Denote $|S| = m$. Label the vertices of $S$:  $1,2, \ldots m$, according to the  order of exposure. That is, vertex $1$ is the first vertex that was exposed, and so on. Denote by $S_i$ the set of vertices  $\{1,2,\ldots, i\}$, and by $Y_i$ the random variable representing the number of neighbors of vertex $i$, which are not in $S_{i-1}$. That is, $Y_i = |N(i)\setminus S_{i-1}|$. The quantity we would like to bound, $|\{(u,v)\in E\ |\ u\in S\}|$, is exactly $\sum_{i=1}^m Y_i$.

From the definition of a $d$-light distribution, conditioned on every possible realization of $S_{i-1}$, $N(S_{i-1})$ and the edges adjacent to $S_{i-1}$, there exists some $d\leq \alpha_i \leq n$ such that $Y_i \sd B(\alpha_i, d/\alpha_i)$. By Lemma~\ref{lemma:lrst}, under the same conditioning, $Y_i \sd Z \sim 2d + B(n^2, \frac{2d}{n^2})$. This implies that conditioned on any realization of $Y_1,\ldots Y_{i-1}$ we have that $Y_i \sd Z$. By Lemma~\ref{lemma:stoch2} we now have that $\sum_{i=1}^m Y_i$ is stochastically dominated by the sum of $m$ independent copies of $Z$. Let $\{X_{i,j}\}_{i\in [m],j\in[n^2]}$ be $n^2\cdot m$ independent Bernoulli random variables such that $Pr[X_{i,j} = 1] = 2d/n^2$ for every $i\in [m]$ and $j\in[n^2]$. Note that for every $i$ we have that  $\dsum_{j=1}^{n^2} X_{i,j} \sim B(n^2, 2d/n^2)$. By the definition of $Z$ we now have that
$$\sum_{i=1}^m Y_i \sd 2dm + \dsum_{i=1}^m \dsum_{j=1}^{n^2} X_{i,j}.$$

By the linearity of expectation, $\expect[\dsum_{i=1}^m \dsum_{j=1}^{n^2} X_{i,j}]=2dm.$ By the the Chernoff bound, there exists a constant $c$ such that when $m\geq c \log{n}$ it holds that $\Pr[\dsum_{i=1}^m \dsum_{j=1}^{n^2} X_{i,j} > 4dm]\leq 1/n^5.$ The proposition follows. 
\end{proof}

\section{Almost $k$-wise Random Orderings}
\label{section:orderings}

Our LCAs will emulate the execution of online algorithms (that are neighborhood dependent as in Definition~\ref{def:NDonline}). One obstacle is that the output of an online algorithm may depend on the order in which it sees the vertices of its input graph (or the edges, in case we are considering an algorithm on edges). As the combined output of an LCA on all vertices has to be consistent, it is important that in all of these executions the algorithm uses the same permutation. Choosing a random permutation on $n$ vertices requires $\Omega(n\log n)$ bits which is disallowed (as it will imply memory $\Omega(n\log n)$). Instead, we would like to have a derandomized choice of a permutation that will be ``good enough" for the sake of our LCAs and can be sampled with as few as $O(\log n)$ bits.

We start with a few definitions and notation. We assume that each vertex has a unique label between $1$ and some $n= \poly(|V|)$. For simplicity, assume $n=|V|$ (or in other words $V=[n]$).
\begin{definition} [Ranking] \label{def:perm} Let $L$ be some positive integer. A function $r: [n] \to [L]$ is a {\em ranking} of $[n]$, where $r(i)$ is called the \emph{level} of $i$. The {\em ordering} $\Pi_r$ which corresponds to $r$ is a permutation on $[n]$ obtained by defining $\Pi_r(i)$ for every $i\in[n]$ to be its position according to the monotone increasing order of the relabeling $i\mapsto (r(i), i)$. In other words, for every $i,j\in[n]$ we have that $\Pi_r(i)<\Pi_r(j)$ if and only if $r(i)<r(j)$ or $r(i)=r(j)$ and $i<j$. The pair $(r(i),i)$ is called the \emph{rank} of $i$.
\end{definition}


Past works on LCAs (e.g., \cite{ARVX11,HMV14,MRVX12,MV13}) have used $k$-wise almost independent orderings to handle the derandomized ordering of vertices. A family of ordering $\mathbf{\Pi}=\{\Pi_r\}_{r\in R}$ , indexed by $R$ is \emph{$k$-wise independent} if for every subset $S\subseteq [n]$ of size $k$, the projection of $\mathbf{\Pi}$ onto $S$ (denoted by $\mathbf{\Pi}(S)$) is uniformly distributed over all $k!$ possible permutations of $S$. Denote this uniform distribution by $U(S)$. We have that $\mathbf{\Pi}$ is \emph{$\eps$-almost $k$-wise independent} if for every $k$-element subset $S$ we have that $\Delta(\mathbf{\Pi}(S), U(S)) \leq \eps.$ One can give adaptive versions of these definitions (in the spirit of Definition~\ref{def:adaptiveKwise}). For simplicity, in this section we concentrate on the static case, but we note that  our results in this section, Lemma~\ref{lem:orderings} and Corollary~\ref{coro:orderings}, extend to the adaptive case. (In Section~\ref{subsection:better}, which contains our main contributions, we focus on the adaptive case.)

It is easy to show that $k$-wise independent functions (or even almost $k$-wise independent functions), directly give $\eps$-almost $k$-wise independent orderings.

\begin{lemma}\label{lem:orderings}
For every $n,L,k \in \N, \eps>0$, such that $k \leq n$ and $L\geq k^2/\eps$, if $\mcH = \{h : [n] \rightarrow [L]\}$ is $k$-wise independent then  $\mathbf{\Pi}=\{\Pi_h\}_{h\in \mcH}$ is a family of $\eps$-almost $k$-wise independent ordering.
\end{lemma}

\begin{proof}
Fix the set $S$. There is probability smaller than $\eps$ on the choice of $h\in \mcH$ that there exist two distinct values $i$ and $j$ in $S$ such that $h(i)=h(j)$. Conditioned on such collision not occurring the order is uniform by the definition of $k$-wise independent hashing.
\end{proof}

Lemma~\ref{lem:orderings} and Theorem~\ref{thm:simple} imply the following.

\begin{corollary}\label{coro:orderings}
For every $n,L,k \in \N, \eps>0$, such that $k \leq n$ and $L\geq k^2/\eps$, there exists a construction of $\eps$-almost $k$-wise independent random ordering over $[n]$ whose seed length is $O(k\log{n})$.
\end{corollary}


This result improves the space bounds (and hence running times) of the algorithms of, for example, \cite{ARVX11,MRVX12,MV13}. For constant $k$, the seed length of the ordering is $O(\log{n})$ and the evaluation time is $O(1)$.
A concrete example of an algorithm whose running time and space bounds can be significantly improved using this construction is the maximal independent set LCA of \cite{ARVX11}, where they prove a space bound of $O(\log^2{n})$ and a time bound of $O(\log^3{n})$). The LCA simulates Luby's algorithm \cite{Luby86} on a graph of constant degree for a constant number of rounds. It can be shown how to improve both the time and space bounds to $O(\log{n})$, based on the construction Corollary~\ref{coro:orderings}.\footnote{Instead of each vertex choosing itself with some probability, and then selecting itself if none of its neighbors is also selected, each vertex selects itself if it is earlier in the ordering than all of its neighbors. The ordering can be generated using Corollary~\ref{coro:orderings}.}

 If $k$ is not a constant, but $O(\log{n})$, as is the case in most algorithms of \cite{ARVX11,MRVX12,MV13}, the bound of Corollary \ref{coro:orderings} gives us a seed length of $O(\log^2{n})$, in comparison to the $O(\log^3{n})$ attained using the almost $k$-wise ordering of \cite{ARVX11}. This could potentially be further improved, but one can observe that a lower bound on the seed length of almost $\log n$-wise independent ordering is $\Omega(\log n\log\log n)$. We now show how to define derandomized orderings that require seed $O(\log n)$ and, while not being $\log n$-wise almost independent, are still sufficiently good for our application.

\section{Ordering with logarithmic seeds}
\label{subsection:better}

Lemma \ref{lem:orderings} and Corollary \ref{coro:orderings} show that for $L=\poly(n)$ when $\mcH = \{h : [n] \rightarrow [L]\}$ is $k$-wise independent, then 
$\mathbf{\Pi}=\{\Pi_h\}_{h\in \mcH}$ is a family of $1/\poly(n)$-almost $k$-wise independent orderings, and that this requires a seed of length $O(k\log{n})$. We now consider what happens if we let $L$ be much smaller, namely a constant. This will reduce the seed length to $O(k+\log n)$ (when we let $\mcH$ be {\em almost} $k$-wise independent). However, we will lose the $k$-wise almost independence (for sub-constant error) of the ordering. Consider two variables $i<j$. Conditioned on $h(i)\neq h(j)$, the order of $i$ and $j$ under $\Pi_h$ is uniform. But with constant probability $h(i)=h(j)$ and then, according to Definition\ref{def:perm}, $i$ will come before $j$ under $\Pi_h$ (recall that the ordering is based on the labels $(h(i),i)$). Nevertheless, even though $\mathbf{\Pi}$ is no longer $k$-wise $1/n$-almost independent, we will show that a constant $L$ suffices for our purposes. To formally define what this means we introduce some definitions.

\begin{definition}[Level, levelhood]\label{def:level}
Let $G=(V,E)$ be a graph. Let $h:V\rightarrow \mbN$ be a function which assigns each vertex an integer. For each $v \in V$, we call $h(v)$ the \emph{level} of $v$.  Denote the restriction of a set of vertices $S \subseteq V$ to only vertices of a certain level $\ell$, $\{x \in S: h(x)=\ell\}$ by $S\|_{\ell}$.

 Let $S \subseteq V$, and let $N_{\ell}(S)$ be the neighbors of $S$ that are in level $\ell$. That is $N_{\ell}(S) = \{u \in N_{\ell}(S) : u \in N(S) , h(u) = \ell \}$.
 The \emph{$\ell^{th}$ levelhood of $S$} (denoted $\Psi_{\ell}(S)$), is defined recursively as follows. $\Psi_{\ell}(\emptyset) = \emptyset$.
$\Psi_{\ell}(S) = \{S \cup  \Psi_{\ell}(N_{\ell}(S))\}$. In other words, we initialize $\Psi_{\ell}(S) =S$, and add to $\Psi_{\ell}(S)$ neighbors of level $\ell$, until we cannot add any more vertices.
\end{definition}

\begin{definition}[Relevant vicinity]\label{def:vicinity}
The \emph{relevant vicinity} of a vertex $v$, denoted $\Im(v)$ (relative to a hash function $h: V \rightarrow [L]$), is defined constructively as follows. Let $\Pi_h$ be the permutation defined by $h$, as in Definition \ref{def:perm}. Initialize $\Im(v) = \{v\}$. For each $u \in \Im(v)$, add to $\Im(v)$ all vertices $w \in N(u) : \Pi_h(w)<\Pi_h(u)$. Continue adding vertices to $\Im(v)$ until the neighbors of all the vertices in $\Im(v)$, appear after them in $\Pi_h$.
\end{definition}
The relevant vicinity of a vertex $v$ is exactly  the vertices that our LCA will query when inquired about $v$.\footnote{In specific cases, better implementations exist that do not need to explore the entire relevant vicinity. For example, once an LCA for maximal independent set  sees that a neighbor of the inquired vertex is in the independent set, it knows that the inquired vertex is not (and can halt the exploration). Nevertheless, in order not to  make any assumptions on the online algorithm, we assume that the algorithm explores the entire relevant vicinity.} Because we can not make any assumptions about the original labeling of the vertices, we upper bound the size of the relevant vicinity by defining the \emph{containing vicinity}, where we assume that the worst case always holds. That is, if two neighbors have the same level,  depending on which is queried, the other appears before it in the permutation. The size of the containing vicinity is clearly an upper bound on the size of a relevant vicinity, for the same hash function.

\begin{definition}[Containing vicinity]\label{def:irvicinity}
The \emph{containing vicinity} of a vertex $v$ (relative to a hash function $h: V \rightarrow [L]$), given that  $h(v) = \ell$, is $\Psi_{L} (\Psi_{L-1} (\ldots \Psi_{\ell+1}(\Psi_l(v))\ldots))$. In other words, let $S_{\ell}$ be the $\ell^{th}$ levelhood of $v$.  Then $S_{\ell+1} = \Psi_{\ell+1}(S_{\ell})$, and so on. The containing vicinity is then $S_L = \Psi_L(S_{L-1})$.
\end{definition}

We wish to bound the size of the relevant vicinity of any vertex:

\subsection{Upper Bounding the Size of the Relevant Vicinity}
\label{section:upper_b}

\begin{lemma}\label{lemma:main}

Let $G=(V,E)$ be a $d$-light graph, $|V|=n$, and let $L$ be an integer such that $L>24d$. Let $c$ be such that Propositions~\ref{prop:1} and  \ref{prop:neighbors} hold, and let $\kappa=2^L c\log{n}$. Let $k = 6d\kappa$, and let $h$ be an adaptive  $k$-wise $\frac{1}{2n^5}$-dependent hash function, $h:V \rightarrow [L]$.  For any vertex $v \in G$, the relevant vicinity of $v$ has size at most $\kappa$ with probability at least $1-\frac{1}{n^4}+\frac{1}{n^5}$.

\end{lemma}

The following claim will help us prove Lemma \ref{lemma:main}.


\begin{clam} \label{lemmaclaim}
Let the conditions of Lemma \ref{lemma:main} hold; assume without loss of generality that $h(v)=1$, and let $S_0 = \{v\}, S_1 = \Psi_1(S_0),  \ldots, S_{\ell+1}= \Psi_{\ell+1}(S_{\ell}),\ldots, S_L = \Psi_{L}(S_{L-1})$. Then for all $0 \leq i \leq L$,
$$\Pr[|S_i|\leq 2^ic\log{n} \wedge |S_{i+1}| \geq 2^{i+1}c \log{n}]\leq \frac{2}{n^5}.$$
\end{clam}
\begin{proof}

Fix $i$ and denote the bad event for $i$ as $B_i$. That is, $B_i  \equiv |S_i|\leq 2^ic\log{n} \wedge |S_{i+1}| \geq 2^{i+1}c \log{n}$.
We make the following observation:
\begin{equation}
B_i \Rightarrow |S_{i+1}| > 2^{i+1}c \log{n} \wedge |S_{i+1}\|_{i+1}|>\frac{|S_{i+1}|}{2}, \label{eq:psibig}
\end{equation}
because $S_{i+1}\|_{i+1} = S_{i+1}\setminus S_i$. In other words, this means that if $B_i$ occurs then the majority of elements in $S_{i+1}$ must have come from the $(i+1)^{th}$-levelhood of $S_i$.


For all $0<i \leq L$ let $T_i = S_i \cup N(S_i)$.
Define two bad events:
\begin{align*} 
B_i^1 &\equiv |S_{i+1}| > 2^{i+1}c \log{n} \wedge |T_{i+1}| > 6d |S_{i+1}| \\ 
B_i^2 &\equiv |S_{i+1}| > 2^{i+1}c \log{n} \wedge |T_{i+1}\|_{i+1}| > \frac{12d}{L}|S_i+1| \\ 
\end{align*}
 From Equation~\eqref{eq:psibig}, the definition of $L$, and the fact that $T_{i+1}\|_{i+1}=S_{i+1}\|_{i+1}$, we get $B_i \Rightarrow B_i^2$.

By Proposition~\ref{prop:neighbors},
\begin{equation}
\Pr[B_i^1] \leq \frac{1}{n^5}\label{eq:fact1},
\end{equation} because  $S_{i+1}$ is an adaptively exposed subset, $T_{i+1} = S_{i+1} \cup N(S_{i+1})$ and the size of $S_{i+1}$ satisfies the conditions of the proposition. Given $|S_{i+1}| > 2^{i+1}c \log{n}$, it holds that $|T_{i+1}| > 2^{i+1}c \log{n}$ because $S_{i+1} \subseteq T_{i+1}$. Also note that $T_{i+1}$ was defined based on the value of $h$ on elements in $T_{i+1}$ only. Therefore, the conditions of Proposition \ref{prop:1} hold for $|T_{i+1}|$.

\begin{align}
\Pr[B_i^2:\neg B_i^1] &\leq \Pr[|T_{i+1}\|_{i+1}| > \frac{12d}{L}|S_{i+1}|: |T_{i+1}|< 6d|S_{i+1}|] \notag\\
&\leq \Pr[|T_{i+1}\|_{i+1}| > \frac{2|T_{i+1}|}{L}] \notag\\
& \leq \frac{1}{n^5}, \label{eq:fact2}
\end{align}
where the last inequality is due to Proposition~\ref{prop:1},

Because
$$\Pr[B_i^2] = \Pr[B_i^2: B_i^1]\Pr[B_i^1]  +  \Pr[B_i^2:\neg B_i^1]\Pr[\neg B_i^1],$$
from Equations \eqref{eq:fact1} and \eqref{eq:fact2}, the claim follows. 
\end{proof}

\begin{proof}[Proof of Lemma \ref{lemma:main}]  To prove Lemma \ref{lemma:main}, we need to show that
\begin{equation*}
\Pr[|S_L|> 2^L c \log{n}]<\frac{1}{n^4}-\frac{1}{n^5},
\end{equation*}
We show that for $ 0\leq i \leq L, \Pr[|S_i| > 2^ic \log{n}]< \frac{2i}{n^5}$, by induction.
For the base of the induction, $|S_0| = 1$, and the claim holds.
For the inductive step, assume that $\Pr[|S_i|> 2^i c\log{n}]< \frac{2i}{L}$.
Then
\begin{align*}
\Pr[|S_{i+1}|> 2^{i+1}c\log{n}] &= \Pr[|S_{i+1}|> 2^{i+1}c\log{n} : |S_i|>2^{i}c\log{n}]\Pr[|S_i|>2^{i}c\log{n}]\\
& + \Pr[|S_{i+1}|> 2^{i+1}c\log{n} : |S_i|\leq2^{i}c\log{n}]\Pr[|S_i|\leq 2^{i}c\log{n}].
\end{align*}
From the inductive step and
Claim~\ref{lemmaclaim}, using the union bound, the lemma follows. 
\end{proof}

From Lemma \ref{lemma:main} and Proposition \ref{prop:neighbors},  we immediately get
\begin{corollary}\label{corr:relevant}
Let $G=(V,E)$ be a $d$-light graph, where $|V|=n$, and let $L$ be an integer such that $L>24d$. Let $c$ be such that Propositions \ref{prop:1} and \ref{prop:neighbors} hold, and let $\kappa=2^L c\log{n}$. Let $k = 6d\kappa$, and let $h$ be an adaptive  $k$-wise $\frac{1}{2n^5}$-dependent hash function.  For any vertex $v \in G$, let $S_L$ be the relevant vicinity of $v$. Then
$$\Pr[|\{(u,v)\in E\ |\ u\in S_L\}|> k]<1/n^4.$$
\end{corollary}

Corollary  \ref{corr:relevant} essentially shows the following: Assume that $G=(V,E)$ is a $d$-light graph, and there is some function $F:V\rightarrow \mcR$ that is computable by a neighborhood-dependent online algorithm $\mcA$. Then, in order to compute $F(v)$ for any vertex $v \in V$, we  only need to look at a logarithmic number of vertices and edges with high probability. Note that in order to calculate the relevant vicinity, we  need to look at all the vertices in the relevant vicinity and all of their neighbors (to make sure that we have not overlooked any vertex). This is upper bound by the number of edges which have an endpoint in the relevant vicinity, as the relevant vicinity is connected. Furthermore, as we will see in Section \ref{section:onlinetolca}, we would like to store the subgraph induced by the relevant vicinity, and for this, we need to store all of the edges.

Applying a union bound over all the vertices gives that the number of queries we need to make per inquiry is $O(\log{n})$ with probability at least $1-1/n^3$ even if we inquire about all the vertices. 

\section{Expected Size of the Relevant Vicinity}
\label{section:expected}
In Section \ref{section:onlinetolca}, we show constructions of the subgraph induced by the relevant vicinity whose running times and space requirements are dependent on $t_v$ and $t_e$, the size of the relevant vicinity and the number of edges adjacent to the relevant vicinity, respectively. All the dependencies can be upper bound by $O((t_e)^2)$. In this section, we prove that the expected value of $(t_e)^2$ in a $d$-light graph is a constant (when $d$ is a constant).

\begin{proposition}\label{prop:expect_simple_paths}
For any $d$-light graph $G=(V,E)$ and any vertex $v \in V$, the expected number of simple paths of length $t$ originating from $v$ is at most $d^t$.
\end{proposition}
We prove a slightly more general  claim, from which Proposition \ref{prop:expect_simple_paths} immediately follows (taking $S$ in the proposition to be the empty set).
\begin{clam}
For any $d$-light graph $G=(V,E)$, any adaptively exposed subset $S \subseteq V$, and any vertex $v \in N(S)\setminus S$ (or any vertex $v$ if $S$ is empty),   the expected number of simple paths of length $t$ originating from $v$ and not intersecting with $S$ is at most $d^t$.
\end{clam}
\begin{proof}
The proof is by induction on $t$. For the base of the induction, $t=0$, and there is a single simple path (the empty path). For the inductive step, let $t>0$ and assume that the claim holds for $t-1$. We show that it holds for $t$.
Given $S$, let $S'=S \cup \{v\}$. Let $a$ be a random variable representing the number of neighbors of $v$ that are not in $S$; that is, $a = |N(v)\setminus S|$. Because $G$ is $d$-light,  $\expect(a) \leq d$.
Fixing $a$, label the neighbors of $v$ that are not in $S$ by $w_1, w_2, \ldots, w_a$. By the inductive hypothesis (as $S'$ is also adaptively exposed), for all $i = 1, \ldots, a$,
the expected number of simple paths of length $t-1$ originating from $w_i$ and not intersecting with $S'$ is at most $d^{t-1}$. The expected number of simple paths of length $t$ originating from $v$ and not intersecting with $S$ is therefore upper bounded by
$$\sum_j \Pr[a=j]j\cdot d^{t-1} = \expect[a] d^{t-1} \leq d^t.$$
\end{proof}

For any simple path $p$ originating at some vertex $v$, we would like to determine whether all the vertices on $p$ are in the relevant vicinity of $v$. As we cannot make any assumptions about the original labels of the vertices, we upper bound this by the probability that the levels of the vertices on the path are non-decreasing.
 \begin{definition}[Legal path]
 We say that a path $p= v\leadsto u$ is \emph{legal} if it is simple and the labels of the vertices on $p$ are in non decreasing order.
 \end{definition}
  \begin{definition}[Prefix-legal path]
  We say that a path $p= v\leadsto u$ of length $t$ is \emph{prefix-legal} if it is simple, and the  prefix of $p$ of length $t-1$ is legal.
  \end{definition}
\begin{proposition}\label{prop:probability_of_a_path}
Let $G=(V,E)$ be a $d$-light graph and let the conditions of Lemma~\ref{lemma:main} hold. That is, $k, \kappa = O(\log{n})$ and $L>24d$. For any $c>0$ there exists a value $L = O(d)$ for which the following holds: Let $h$ be an adaptive  $k$-wise $\frac{1}{2n^5}$-dependent hash function, $h:V \rightarrow [L]$.
For any path $p$ of length $t'<k-1$ originating at some vertex $v$, the probability that $p$ is legal is at most $d^{-ct'}$.
\end{proposition}

\begin{proof}

For any simple  path $p$ of length ${t'}$ from $v$, there are $|L|^{t'}$ possible values for the levels of the vertices of $p$.  Let the values be $r_0, r_1, \ldots, r_{{t'}-1}$. We define ${t'}+1$ new variables $a_0 = r_0$, $a_1 = r_1-r_0$, \ldots, $a_{{t'}-1} = r_{{t'}-1}-r_{{t'}-2}$, $a_{{t'}}= L-r_{{t'}-1}$. Clearly, the $a_i$'s uniquely define the $r_i$'s and $p$ is legal if and only if $a_0, \ldots, a_{{t'}}$ are all non-negative.

Note that the $\dsum_{i=0}^{t'} a_i = L$; hence computing the number of possible legal values of $a_0, \ldots, a_{{t'}}$ is the same as computing the number of ways of placing L identical balls in ${t'}+1$ distinct bins,\footnote{Alternatively, one can view this as placing ${t'}+1$ separators between $L$ balls.}  where each bin represents a vertex and if there are $k$ balls in bin $y$, then $r_y = r_{y-1}+k$. This is known to be ${{t'}+L \choose {t'}}$. 
Therefore, assuming the choices of the levels are all uniform,
$$\Pr[p \text{ is legal}] = \frac{1}{L^T}{{t'}+L \choose {t'}} \leq \left( \frac{e({t'}+L)}{L{t'}}\right)^{t'}.$$

From the definition of $\eps$-almost adaptive $k$-wise independent hash functions, we immediately get
$$\Pr[p \text{ is legal}] \leq \left( \frac{e({t'}+L)}{L{t'}}\right)^{t'} + \eps.$$
The result follows by selecting an appropriate value for $L$.
\end{proof}

The following corollary is immediate, setting $t'=t-1$ in  Proposition~\ref{prop:probability_of_a_path}.
\begin{corollary}\label{corr:probability_of_a_path}
Let $G=(V,E)$ be a $d$-light graph and let the conditions of Lemma~\ref{lemma:main} hold. That is, $k, \kappa = O(\log{n})$ and $L>24d$. For any $c>0$ there exists a value $L = O(d)$ for which the following holds: Let $h$ be an adaptive  $k$-wise $\frac{1}{2n^5}$-dependent hash function, $h:V \rightarrow [L]$.
For any path $p$ of length $t<k$ originating at some vertex $v$, the probability that $p$ is  prefix-legal is at most $d^{c-ct}$.
\end{corollary}


As a warm-up, we first show that the expected of edges adjacent to a relevant vicinity is a constant.

\begin{lemma}\label{lemma:expected}
Let $G=(V,E)$ be a $d$-light graph and let the conditions of Proposition~\ref{prop:probability_of_a_path} hold. That is, $k, \kappa = O(\log{n})$ and $L=O(d)$.  Let $h$ be an adaptive  $k$-wise $\frac{1}{2n^5}$-dependent hash function, $h:V \rightarrow [L]$. Then the expected number of edges adjacent to a relevant vicinity in $G$ is $O(1)$.
\end{lemma}
\begin{proof}
Let $v$ be any vertex, let $S_L$ be the relevant vicinity of $v$, and let $E_L$ be the set of edges with at least one endpoint in $S_L$.
Let $(u,w) \in E$ be any edge , and denote by $p_{(v,u,w)}^{\leq k}$ the indicator random variable whose value is $1$ if there exists a prefix-legal path of length at most $k$ from $v$ to $w$ whose last edge is $(u,w)$, and $0$ otherwise. Similarly,  denote by $p_{(v,u,w)}^{>k}$ the random variable whose value is $1$ if there exists a prefix-legal path of length greater than $k$ from $v$ to $w$ whose last edge is $(u,w)$, and $0$ otherwise.
For any $(u,w) \in E$,
$$\Pr[(u,w) \in E_L] \leq \Pr[p_{(v,u,w)}^{\leq k}] + \Pr[p_{(v,u,w)}^{>k}].$$

Therefore,
\begin{align}
\expect[|E_L|] &\leq \dsum_{(u,w) \in E} \Pr[(u,w) \in E_L] \notag\\
&\leq \dsum_{(u,w) \in E} \Pr[p_{(v,u,w)}^{\leq k}] + \dsum_{(u,w) \in E} \Pr[p_{(v,u,w)}^{>k}] \notag \\
&\leq \dsum_{(u,w) \in E} \Pr[p_{(v,u,w)}^{\leq k}] + n^2 \Pr[|S_L| > k-1] \notag \\
&\leq \dsum_{(u,w) \in E} \Pr[p_{(v,u,w)}^{\leq k}] +1 \label{eq:lemma5}\\
&\leq \dsum_{\text{paths of length } \leq k\ \text{from } v} \Pr[\text{path is prefix-legal }] +1 \notag\\
&\leq \dsum_{t \leq k} \left( \dsum_{\text{paths of length}\ t\ \text{from } v} \Pr[\text{path is prefix-legal }]\right)  +1\notag \\
&\leq \dsum_{t \leq k} d^{t+c-ct} +1 = O(1) \label{eq:from2props}
\end{align}
where Inequality \eqref{eq:lemma5} is due to Lemma~\ref{lemma:main}, and Inequality \eqref{eq:from2props} is due to Proposition \ref{prop:expect_simple_paths} and Corollary \ref{corr:probability_of_a_path}.
\end{proof}

\begin{lemma}\label{lemma:expected_squared}
Let $G=(V,E)$ be a $d$-light graph and let the conditions of Proposition~\ref{prop:probability_of_a_path} hold. That is, $k, \kappa = O(\log{n})$ and $L=O(d)$.  Let $h$ be an adaptive  $k$-wise $\frac{1}{2n^5}$-dependent hash function, $h:V \rightarrow [L]$. Denote by $E_L$ the number of edges that have at least one endpoint in the relevant vicinity of some vertex $v$. Then, $\expect [|E_L|^2] = O(1)$.
\end{lemma}
\begin{proof}
For any edge $e=(u,w) \in E$, let $\mathbb{I}_e$ be an indicator variable whose value is $1$ if $e \in E_L$ and $0$ otherwise. Let $\mathbb{I}_{|S_L| > k}$ be an indicator variable whose value is $1$ if $|S_L|>k$ and $0$ otherwise.
Then
$$|E_L|^2 = \left( \dsum_{e \in E}\mathbb{I}_e\right) ^2 =  \dsum_{e \in E}\mathbb{I}_e\dsum_{f \in E}\mathbb{I}_f$$
Let $p_{(v,u,w)}^{\leq k}$ and $p_{(v,u,w)}^{>k}$ be as in the proof of Lemma~\ref{lemma:expected}.
\begin{align}
|E_L|^2 &\leq  \dsum_{(u,w) \in E} (p_{(v,u,w)}^{\leq k} + p_{(v,u,w)}^{>k})   \dsum_{(x,y) \in E} ( p_{(v,x,y)}^{\leq k} + p_{(v,x,y)}^{>k})   \notag \\
&= \dsum_{(u,w) \in E}\dsum_{(x,y) \in E}  (p_{(v,u,w)}^{\leq k} + p_{(v,u,w)}^{>k})  (p_{(v,x,y)}^{\leq k} + p_{(v,x,y)}^{>k}) \notag \\
&\leq \dsum_{(u,w) \in E}\dsum_{(x,y) \in E} (p_{(v,u,w)}^{\leq k}p_{(v,x,y)}^{\leq k} +3\mathbb{I}_{|S_L| > k-1})  \notag  \\
&\leq 3n^4\mathbb{I}_{|S_L| > k-1} + \dsum_{(u,w) \in E}\dsum_{(x,y) \in E}(p_{(v,u,w)}^{\leq k}p_{(v,x,y)}^{\leq k}).  \label{eq:lemm5}
\end{align}
For every vertex $u \in V$, let $\sigma_u^t$ denote the number of simple paths from $v$ to $u$ of length $t$, and label these paths arbitrarily by $q_u^t(i), i=1,2,\ldots \sigma_u^t$. For each path $q_u^t(i)$, let $\hat{q}_u^t(i)$ be the random variable whose value is $1$ if $q_u^t(i)$ is prefix-legal, and $0$ otherwise. Let $\Lambda_v^t$ denote the total number of simple paths of length $t$ originating in $v$.

\begin{align}
 \dsum_{(u,w) \in E}\dsum_{(x,y) \in E} (p_{(v,u,w)}^{\leq k}p_{(v,x,y)}^{\leq k})   &\leq \dsum_{w \in V}\dsum_{t\leq k}\dsum_{i=1}^{\sigma_w^t}   \dsum_{y \in V}\dsum_{s\leq k}\dsum_{j=1}^{\sigma_y^s}  \hat{q}_w^t(i)  \hat{q}_y^s(j)  \notag\\
 &\leq 2\dsum_{w \in V}\dsum_{t\leq k}\dsum_{i=1}^{\sigma_w^t}   \dsum_{y \in V}\dsum_{s\leq t}\dsum_{j=1}^{\sigma_y^s} \hat{q}_w^t(i)  \hat{q}_y^s(j) \label{eq:3}\\
&= 2 \left( \dsum_{w \in V}\dsum_{t\leq k}\dsum_{i=1}^{\sigma_w^t} \hat{q}_w^t(i)\right)  \left(   \dsum_{y \in V}\dsum_{s\leq t}\dsum_{j=1}^{\sigma_y^s}\hat{q}_y^s(j) \right) \notag\\
&\leq 2 \dsum_{w \in V}\dsum_{t\leq k}\dsum_{i=1}^{\sigma_w^t} \hat{q}_w^t(i) \dsum_{s \leq t}\Lambda_y^s , \notag
\end{align}
where Inequality \eqref{eq:3} is because we order the paths by length and either path can be longer. 
\begin{align}
\expect[ \dsum_{(u,w) \in E}\dsum_{(x,y) \in E} (p_{(v,u,w)}^{\leq k}p_{(v,x,y)}^{\leq k})]
&\leq 2\dsum_{t \leq k}\dsum_{\text{paths of length }  t\ \text{from } v} \Pr[\text{path is prefix-legal}] \dsum_{s \leq t}\Lambda_y^s  \notag\\
&\leq  2\dsum_{t \leq k} d^{t+c-ct} \dsum_{s \leq t} \Lambda_y^s \notag\\
&\leq 2 \dsum_{t \leq k} d^{t+c-ct} d^{2t} \notag \\
&\leq 2 \dsum_{t \leq k} d^{3t+c-ct} = O(1).\label{eq:last}
\end{align}
From Corollary~\ref{corr:probability_of_a_path} we know that we can choose an $L=O(d)$ such that Inequality~\eqref{eq:last} holds. From Inequalities~\eqref{eq:lemm5} and \eqref{eq:last},  Lemma~\ref{lemma:main} and the linearity of expectation, the lemma follows. 
\end{proof}

\section{From Online to LCA}
\label{section:onlinetolca}
Let $G=(V,E)$ be a $d$-light graph, $d>0$.  Let $F$ be a search problem on $V$, and assume that there exists a neighborhood-dependent online algorithm $\mcA$ for $F$. We show how we can use the results of the previous sections to construct an LCA for $F$.
Given an inquiry $v \in V$, we would like to generate a permutation $\Pi$ on $V$, build the relevant vicinity of $v$ relative to $\Pi$, and simulate $\mcA$ on these vertices in the order of $\Pi$ . Because $\mcA$ is neighborhood-dependent, we do not need to look at any vertices outside the relevant vicinity  in order to correctly compute the output of $\mcA$ on $v$, and so our LCA will output a reply consistent with the execution of $\mcA$ on the vertices, if they arrive according to $\Pi$. In order to simulate $\mcA$ on the correct order, we need to store the relevant vicinity and label the vertices in a way that defines the ordering. We show two ways of doing this. The first gives a better time bound, at the expense of a worse space bound. The second gives a better space bound, at the expense of a worse time bound. It remains an open problem whether we can achieve ``the best of both worlds''- an LCA requiring $O(\log{n}\loglog{n})$ time and space (or even $O(\log n)$). We note that in expectation, both our LCAs require $O(\loglog{n})$ time and $O(\log{n})$ space.

\begin{definition}
We say that online algorithm $\mcA$ is \emph{crisp} if $\mcA$ requires time and space linear in its input length (where the input length is measured by the number of words) per query, and the output of $\mcA$ is $O(1)$ per query.
\end{definition}
Most of the algorithms that we wish to convert to LCAs using the techniques of this paper are indeed crisp; for example the greedy algorithm for maximal independent set requires computation time and space linear in the number of neighbors of each vertex, and the output per vertex is a single bit. Not all the algorithms we wish to handle are necessarily crisp; in the case of vertex coloring, the output of the greedy algorithm can be $\log{\Delta}$, where $\Delta$ is the maximal degree of the graph. Nevertheless, as to avoid a cumbersome statement of our results, we restrict ourselves to crisp algorithms; it is straightforward to extend our results to non-crisp algorithms (as we remark below).

\begin{theorem}\label{thm:main}
Let $G=(V,E)$ be a $d$-light graph, where $d>0$ is a constant.  Let $F$ be a neighborhood-dependent search problem on $V$. Assume $\mcA$ is a crisp neighborhood-dependent online algorithm that  correctly computes $F$ on any order of arrival of the vertices. Then
\begin{enumerate}
\item There is an $(O(\log{n}\loglog{n}),  O(\log^2{n}), 1/n)-$ LCA for $F$.
\item There is an $(O(\log^2{n}), O(\log{n}\loglog{n}), 1/n)-$ LCA for $F$.
\end{enumerate}
Furthermore, both LCAs require, in expectation, $O(\loglog{n})$ time and $O(\log{n})$ space.
\end{theorem}

\begin{proof}

We show two methods of constructing an LCA from the  online algorithm $\mcA$. In both, given a query $v \in V$, we use adjacency lists to store the containing vicinity, $\Psi(v)$.  We use a single bit to indicate for each vertex, whether it is in the relevant vicinity, $\Im(v)$. This means that for each vertex in $\Im(v)$, we keep a list of all of its neighbors, but for vertices that are in $\Psi(v)\setminus \Im(v)$, we don't need to keep such a list. We denote the number of vertices in the relevant vicinity, $|\Im(v)|$, by $t_v$, and the total number of edges stored by $t_e$. Note that $t_e \geq |\Psi(v)|$.
This adjacency list representation of $\Psi(v)$ is generated slightly differently in the two constructions. In both cases we label this data structure by $D(v)$. For clarity, we abuse the notation, and use the same name, $u$, for $u \in V$, and for  the vertex which represents $u$ in $D(v)$.
\begin{meth}\label{meth1}
The first time the LCA is invoked, it chooses a random function $h$ from a family of adaptive $k$-wise $\frac{1}{2n^5}$-dependent hash functions as in  Theorem~\ref{thm:almost}, Lemma~\ref{lemma:main} and Proposition~\ref{prop:probability_of_a_path}. The LCA receives as an inquiry a vertex $v$, and computes $h(v)$. It then discovers the relevant vicinity using DFS. For each vertex $u$ that it encounters, it relabels the vertex $(h(u),u)$. The LCA simulates $\mcA$ on the vertices arriving in the order induced by the new labels.
\end{meth}

The size of the new label for any vertex $u$ is $|(h(u),u)| = O(\log{n})$. In addition, we need to store $\mcA(u)$ for every vertex $u \in \Im(v)$ (which is $O(1)$ because we assume $\mcA$ is crisp). Overall, because $|(h(u),u)| = \log{n}$, the space required for the LCA is upper bounded by $O((t_e + t_v)\log{n} + t_v)$.  From Corollary~\ref{corr:relevant}, $t_v, t_e = O(\log{n})$. In expectation, by Lemma~\ref{lemma:expected_squared}, $\expect[t_e + t_v] = O(1)$. This gives us the required space bounds.
 To analyze the running time, we make the following observation.
 \begin{obs}\label{obs:access1}
 In Construction~\ref{meth1}, given $u \in D(v)$, we can access $u \in V$ in $O(1)$.
 \end{obs}
  We look at each stage of the construction separately:
\begin{enumerate}[noitemsep,nolistsep]
\item Constructing $D(v)$ is done by DFS, which takes time $O(t_v + t_e)$, as well as the time it takes to generate $|\Im(v)|$ labels, which requires invoking $h$ at most $t_e$ times, and, by Theorem~\ref{thm:almost},  this requires $O(\loglog{n})$ time per label.
\item Sorting the labels takes $O(t_v \log t_v)$.
\item Simulating $\mcA$ on $\Psi(v)$ now takes $O(t_e)$ (since $\mcA$ is crisp).
\end{enumerate}
Given the high probability upper bounds and expected values of $t_v$ and $t_e$, the first part of the theorem follows.
(Note that if $\mcA$ is not crisp in the sense that computing $F(v)$ is more than linear in the number of neighbors of $v$, this must be taken into account in the running time.)

In the first construction method, we give each vertex a label of length $O(\log{n})$. This seems wasteful, considering we know that the expected size of the relevant vicinity is $O(1)$, and  that its size is  $O(\log{n})$ w.h.p. We therefore give a more space-efficient method of constructing the induced subgraph.

\begin{meth}\label{meth2}
As in the first method, the first time the LCA is invoked, it chooses a random function $h$ from a family of adaptive $k$-wise $\frac{1}{2n^5}$-dependent hash functions as in  Theorem~\ref{thm:almost}, Lemma~\ref{lemma:main} and Proposition~\ref{prop:probability_of_a_path}. Again, we would like to construct the induced subgraph of the relevant vicinity, but to save memory we will not hold the original labels of the vertices (which require $\log n$ bits to represent), but rather new labels that require at most $\loglog n$ bits to describe (logarithmic in the size of the relevant vicinity). As before, the LCA receives as an inquiry a vertex $v$, and computes $h(v)$. We still use $(h(v),v)$ to determine the ordering, however we do not commit this ranking to memory. We  initialize $S = \{v\}$, and give $v$ the label $1$. In each round $i$, we look at $N(S)$, and choose the vertex $u$ with the highest rank $(h(u),u)$, out of all the vertices in $N(S)$ which have a lower rank than their neighbor in $S$. We then add $u$ to $S$, and give it the label $i+1$.  When we have discovered the entire relevant vicinity, we simulate $\mcA$ on the vertices in the reverse order of the new labels.
\end{meth}
\paragraph{Note on the required data structure:} To efficiently build $S$, we need to use a slightly different data structure used for storing the adjacency lists than in Construction~\ref{meth1}; in fact, we have two adjacency list data structures. The first, $D_1(v)$, contains only the vertices in the relevant vicinity (not vertices in $\Psi(v) \setminus \Im(v)$). The second, $D_2(v)$, holds the neighbors of the vertices of $S$ which have not yet been added to $S$.  In $D_1(v)$, each vertex is represented by its new label. In $D_2(v)$, each vertex is represented only by its level. We need $D_2$ to avoid recalculating $h(u)$ more than once for each vertex $u$. In both $D_1(v)$ and $D_2(v)$, for each edge $(i,j)$ which represents the edge $(v_i, v_j)$, we also store the position of this edge among the edges that leave $v_i$, and its direction of discovery.

\paragraph{Correctness:} Note that $v$ is the vertex of highest rank in its relevant vicinity, and indeed it holds by induction that at step $i$, the subgraph we expose will contain vertices $v_1=v,v_2,\ldots v_i$ that have the highest ranking in the relevant vicinity of $v$ (and such that $v_i$ has the $i$th highest rank). This guarantees the correctness of $\mcA$ - the reverse order of the labels is exactly the correct ranking of the vertices of the relevant vicinity.

\paragraph{Complexity:}  The size of the new label for any vertex $u$ is  $O(\log{t_e})$. In addition, we need to store, for each edge, its position relative to the edges, the edge's direction of discovery, and $\mcA(u)$ for every vertex $u \in \Im(v)$. Because the graph is $d$-light, we know that the degree of each edge is $O(\log{n})$ w.h.p., and so keeping the relative position of each edge will require $O(\loglog{n})$ bits w.h.p. Overall the space required for $D_1(v)$ is upper bound by $O((t_e + t_v)(\log t_e) + t_v + t_e\loglog{n})$. The space required for $D_2(v)$ is $O(t_e \cdot |L|) = O(t_e)$. From Corollary~\ref{corr:relevant}, and Lemma~\ref{lemma:expected_squared}, we have the required space bounds. The expected space bound is due to the length of the seed, $O(\log{n})$.

 \begin{obs}\label{obs:access2}
 In Construction~\ref{meth2}, given $u \in D_1(v)\cup D_2(v)$, we can access $u \in V$ in $O(t_v)$.
 \end{obs}
 \begin{proof}
 Given $u \in D_1(v)$ (or $D_2(v)$), we find $v$ by DFS from $u$. As the edges are directed, and $D_1(v)$ and $D_2(v)$ are acyclic, this takes $O(t_v)$. Note that the space required for this DFS may be as much as $t_v\loglog{n}$, but we use that amount of space regardless. We can store the path $v \leadsto u$ using the relative locations of the edges, and follow this path on $G$ to find $u$.
 \end{proof}
  Again, we look at each stage of the construction separately:
\begin{enumerate}[noitemsep,nolistsep]
\item Before we add a vertex to $S$, we need find the vertex $u$ with the lowest $(h(u),u)$ among all vertices in $\Psi(v)\setminus \Im(v)$. This is done by going over all of these vertices to find the minimum, using a DFS on $G$ and the subgraph concurrently, which takes $O(t_e)$. \label{1}
\item Once we have chosen which vertex $u$ to add to $S$, we update  $D_1(v)$ and $D_2(v)$ to include it. When we look at $u$'s neighbors, though, we don't know whether they are already in $D_1(v)$ or $D_2(v)$, as we don't have  pointers to the original vertices in $G$. From Observation~\ref{obs:access2}, though, finding this out takes $O(t_v)$ per neighbor, and updating $D_1(v)$ and $D_2(v)$ takes a further $O(1)$ per neighbor. \label{2}
\item Because of $D_2(v)$, we only need to generate $h$ once for each vertex in $\Psi(v)$. This takes $O(t_e $ $\loglog{n})$. \label{3}
\item Reversing the order of the labels takes $O(t_v)$.\label{4}
\item Simulating $\mcA$ on $\Psi(v)$ takes $O(t_e)$.\label{5}
\end{enumerate}
\ref{3}, \ref{4} and \ref{5} require $O(t_e \loglog{n})$ time in total. \ref{1} accounts for $O((t_e)^2)$, and \ref{2}  for $O(t_e t_v)$ overall.  Lemma~\ref{lemma:expected_squared} gives the required expected time bound.
 Note that we have not discounted the possibility that $\mcA$ requires the original labels (or ``names'') of the vertices, in order to compute $F(v)$. When we encounter a vertex, we can always give $\mcA$ the name of the vertex by exploring the original representation of the graph (the additional time needed is bounded by the time already invested).

The worst case (w.h.p.) running time and space of the LCA are  $O(\log^2{n})$ and $O(\log{n}\loglog{n})$ respectively. (Note that if $\mcA$ is not crisp in the sense that  $|F(v)|$ is not a constant, this must be taken into account in the space bounds.)
\end{proof}
Our results immediately extend to the case that $d=O(\loglog{n})$:
\begin{theorem}\label{thm:loglog}
Let $G=(V,E)$ be a $d$-light graph, where $d=O(\loglog{n})$.  Let $F$ be a neighborhood-dependent search problem on $V$. Assume $A$ is a crisp neighborhood-dependent online algorithm that  correctly computes $F$ on any order of arrival of the vertices. Then there is an $O(\polylog{n}),  $ $ O(\polylog{n}),$ $ 1/n)-$ LCA for $F$.
\end{theorem}

In Section~\ref{section:upper} we show that the techniques of this paper do not hold for graphs where the expected degree is $\omega(\loglog{n})$, and so our results are, in this sense, tight.

\section{Tightness with Respect to $d$-light Graphs}
\label{section:upper}
Our results hold for $d$-light graphs where $d = O(\log\log{n})$. Although we do not discount the possibility that LCAs exist for higher degree graphs, we show that at least, using the technique of simulating an online algorithm on a random ordering of the vertices, we cannot do better.
We do this by showing that the expected relevant vicinity of the root of a complete binary tree of a $d$-regular graph is $\Omega (2^{d/2})$, and hence for $d=\omega(\log\log{n})$, the expected size will be super-polylogarithmic.
\begin{lemma}
Let $T$ be a complete $d$-regular binary tree rooted at $v$, and let $\Pi$ be a uniformly random permutation on the vertices. The expected size of the relevant vicinity of $v$ relative to $\Pi$ is at least $2^{d/2}$.
\end{lemma}

\begin{proof}
Let $X_{\ell}$ be a random variable for the number of vertices on level $\ell$ of the tree that are in the relevant vicinity.
$$\expect[X_{\ell}] = \frac{d^{\ell}}{\ell!} \geq \left( \frac{d}{\ell}\right) ^{\ell}$$
Taking $\ell=d/2$ gives that $\expect[X_{\ell}] \geq 2^{d/2}$.
\end{proof}

\bibliographystyle{plain}
\bibliography{Vardi_PhD_Bibliography}

\appendix

\section{Stochastic Dominance of Binomial Distributions}
\label{app:lrst}
We prove the lemmas stated in Section \ref{section:stoch}.

\begin{lemma} \label{lem:stoch}
Let $Y_1, Y_2$ be independent discrete random variables. Let  $X_1, X_2$ be (possibly) dependent discrete random variables, such that $X_1 \sd Y_1$, and conditioned on any realization of $X_1$, it holds that $X_2 \sd Y_2$, then
$$X_1 + X_2 \sd Y_1 + Y_2.$$
\end{lemma}
\begin{proof}

For every realization $x$ of $X_1$, define a different random variable for $X_2$. That is $X_2(x) = X_2|X_1=x$. Note $\forall x, X_2(x) \sd Y_2$. Further note that $X_1$ and $Y_2$ are independent. From the law of total probability,

\begin{align*}
Pr[X_1 + X_2 >k] & = \dsum_{x}Pr[x+X_2(x) >k] \cdot Pr[X_1=x]\\
&= \dsum_{x}Pr[X_2(x) > k-x]\cdot Pr [X_1=x]\\
&\leq \dsum_{x} Pr[Y_2 > k-x]\cdot Pr [X_1=x]\\
&= \dsum_{x}Pr[x + Y_2 >k]\cdot Pr [X_1=x]\\
&= Pr[X_1 + Y_2 >k]\\
&\leq Pr[Y_1 + Y_2 >k].
\end{align*} 
\end{proof}

This implies
{
\renewcommand{\thetheorem}{\ref{lemma:stoch2}}
\begin{lemma}
Let $\{Y_1, Y_2, \ldots, Y_N\}$  be a sequence of independent random variables. Let  $\{X_1, X_2,\ldots, X_N\}$ be a series of (possibly) dependent random variables. If it holds that $X_1 \sd Y_1$, and for any $1\leq i \leq N$, conditioned on any realization of $X_1,\ldots X_{i-1}$, it holds that $X_i \sd Y_i$, then
 $$\dsum_{j=1}^{N} X_j \sd \dsum_{j=1}^{N}Y_j.$$
\end{lemma}
\addtocounter{theorem}{-1}
}

\begin{proof}
We prove by induction on $i$ that $\sum_{j=1}^{i} X_j \sd \sum_{j=1}^{i} Y_j$. For $i=1$ we have that $X_1 \sd Y_1$. Assume that $\sum_{j=1}^{i} X_j \sd \sum_{j=1}^{i} Y_j$, we prove that $\sum_{j=1}^{i+1} X_j \sd \sum_{j=1}^{i+1} Y_j$. Let $Z_1 = \sum_{j=1}^{i} X_j$, $Z_2 = X_{i+1}$, $W_1 = \sum_{j=1}^{i} Y_j$, and $W_2 = Y_{i+1}$. Applying Lemma~\ref{lem:stoch} with $Z_1, Z_2, W_1, W_2$ we get that $Z_1 + Z_2 \sd W_1 + W_2$ as required. 
\end{proof}

For the second lemma we prove, we require the following well-known inequalities:
\begin{fact}For every $0<x < 1$ and every $y>0$,
$$\left(1-\frac{x}{y}\right) ^y  < e^{-x} < 1-\frac{x}{2}.$$
\end{fact}

\begin{clam}\label{clam:1}
Let $2d <   \alpha \leq n$. Let $X$ and $Y$ be random variables such that $X \sim B(1, \frac{d}{\alpha})$ and $Y \sim B(\lceil\frac{n^2}{\alpha}\rceil, \frac{2d}{n^2})$. Then $X \sd Y$.
\end{clam}
\begin{proof}
$X$ is in fact a random variable with the Bernoulli distribution. As $X$ can only take the values $0$ or $1$, to show stochastic dominance, it suffices to show that $\Pr[Y=0] \leq \Pr[X=0]$.

$$\Pr[Y=0] = \left( 1- \frac{2d}{n^2}\right)^{ \lceil \frac{n^2}{\alpha}\rceil} < e^{-2d/\alpha} <  1-\frac{d}{\alpha} = \Pr[X=0].$$
\end{proof}

\begin{clam}\label{clam:2}
Let $X$  be a random variable such that $X \sim B(\alpha, \frac{d}{\alpha})$, and let $X_1, X_2, \ldots, X_{\alpha-1}$ be random variables duch that $\forall i, X_i \sim  B(1, \frac{d}{\alpha})$. Then  $X \sd \dsum_{i=1}^{\alpha-1} X_i +1$.
\end{clam}
The proof is immediate from the definition of the binomial distribution.

\begin{clam}\label{clam:3}
Let  $Y$ be a random variable such that $Y \sim B(n^2, \frac{2d}{n^2})$ and let $Y_1, Y_2, \ldots, Y_{\alpha-1}$ be random variables such that $\forall i, Y_i \sim B(\lceil\frac{n^2}{\alpha}\rceil, \frac{2d}{n^2})$. Then  $\dsum_{i=1}^{\alpha-1} Y_i \sd Y$.
\end{clam}
\begin{proof}
It suffices to show that $(\alpha-1)\lceil\frac{n^2}{\alpha}\rceil \leq n^2$.
\begin{align*}
(\alpha-1)\left\lceil\frac{n^2}{\alpha}\right\rceil \leq (\alpha-1)(\frac{n^2}{\alpha} +1) \leq n^2 - n + \alpha -1 < n^2,
\end{align*}
because $\alpha \leq n$. 
\end{proof}

Combining Claims~\ref{clam:1}, \ref{clam:2} and \ref{clam:3}, we get

{
\renewcommand{\thetheorem}{\ref{lemma:lrst}}
\begin{lemma}
Let  $Z$ and $X$ be random variables such that $Z \sim 2d + B(n^2, \frac{2d}{n^2})$ and $X \sim B(\alpha, \frac{d}{\alpha})$, where $d \leq \alpha \leq n$. Then $X \sd Z$.
\end{lemma}
\addtocounter{theorem}{-1}
}
\begin{proof}
If $\alpha \leq 2d$, the lemma holds immediately. Assume $\alpha \geq 2d$. Then, using the notation of Claims~\ref{clam:2} and \ref{clam:3}
$$X \sd \dsum_{i=1}^{\alpha-1} X_i +1 \sd \dsum_{i=1}^{\alpha-1} Y_i +1 \sd  B(n^2, \frac{2d}{n^2}) +1 \sd Y+1\sd Z.$$
\end{proof}

\end{document}